\documentclass[12pt,preprint]{elsarticle}
\usepackage{hyperref}
\usepackage{amsthm}
\usepackage{graphicx}
\usepackage{extarrows}
\usepackage{amsthm}
\usepackage{amssymb}
\allowdisplaybreaks

\newcommand{\F}{\mathbb F}

\newcommand{\Tr}{{\rm{Tr}}}

\newtheorem{Thm}{Theorem}[section]
\newtheorem{Lemma}{Lemma}[section]

\newtheorem{Remark}{Remark}[section]

\newtheorem{Corollary}{Corollary}[section]

\journal{Finite Fields Appl.}

\begin{document}

\begin{frontmatter}
\title{A general construction of permutation polynomials of the form $ (x^{2^m}+x+\delta)^{i(2^m-1)+1}+x$ over $\F_{2^{2m}}$}

\author[rvt]{Libo Wang}

\author[focal]{Baofeng Wu\corref{cor1}}
\ead{wubaofeng@iie.ac.cn}

\cortext[cor1]{Corresponding author}


\address[rvt]{College of Information Science and Technology, Jinan University, Guangzhou 510632, China}

\address[focal]{State Key Laboratory of Information Security, Institute of Information Engineering, Chinese Academy of Sciences, Beijing 100093, China}

\begin{abstract}
	Recently, there has been a lot of work on constructions of permutation polynomials of the form $(x^{2^m}+x+\delta)^{s}+x$
	over the finite field $\F_{2^{2m}}$, especially in the case when $s$ is of the form $s=i(2^m-1)+1$ (Niho exponent).
In this paper, we further investigate permutation polynomials with this form. Instead of seeking for sporadic constructions of the parameter $i$,  we give a general sufficient condition on $i$ such that $(x^{2^m}+x+\delta)^{i(2^m-1)+1}+x$ permutes $\F_{2^{2m}}$, that is,   $(2^k+1)i \equiv 1 ~\textrm{or}~ 2^k~(\textrm{mod}~ 2^m+1)$,
where $1 \leq k \leq m-1$ is any integer. This generalizes a recent result obtained by Gupta and Sharma who actually dealt with the case $k=2$. It turns out that most of previous constructions of the parameter $i$ are covered by our result, and it yields many new classes of permutation polynomials as well.
\end{abstract}


\end{frontmatter}

\section{Introduction}
\label{sec1}
Let $q$ be  a power of a prime $p$, and ${\F}_q$  be the finite field with $q$ elements.
A polynomial $f(x) \in {\F}_q[x]$ is called a  $permutation$  $polynomial$ (PP) if its associated
polynomial mapping $f: c \mapsto f(c) $  from ${\F}_q$ to itself is a bijection. PPs over
finite fields have important applications in cryptography, coding and combinatorial design. Classical results on properties, constructions and applications of PPs may be found in \cite{Lidl,Gullen}.
For some recent advances and contributions to this area, we refer to \cite{Hou,Gullen} and the references therein.

Helleseth and Zinoviev \cite{Hellesseth} first investigated PPs of the form
\[\left(\frac{1}{x^2+x+\delta}\right)^{2^{\ell}}+x\]
for the goal of deriving new identities on Kloosterman  sums over $\F_{2^n}$, where $\delta \in \F_{2^n}$ with $\Tr_1^n(\delta)=1$,
and $\ell=0$ or $1$.
This motivated Yuan and Ding \cite{YuanJ1}, Yuan, Ding ,Wang and Pieprzyk \cite{YuanJ2} to investigate the  permutation behavior
of polynomials having the form
\[(x^{p^k} - x+\delta)^s+L(x)\]
over  ${\F}_{p^n}$,  where $k,~s$ are integers, $\delta  \in {\F}_{p^n}$ and $L(x)$ is a linearized polynomial.
An extension  of the above work and some new classes of PPs were found in
\cite{Li,Yuan1,Zeng,Zha1}. Specially, Tu et al. \cite{Tu1}  proposed two classes of PPs over
$\F_{2^{2m}}$ of the form
\begin{eqnarray}
\label{1.01}
(x^{2^m} + x+\delta)^s+x
\end{eqnarray}
for some $s$ satisfies either
\[s(2^m+1) \equiv 2^m+1 ~(\textrm{mod} ~2^{2m}-1)\]
or
\[s(2^m-1) \equiv 2^m-1 ~(\textrm{mod} ~2^{2m}-1).\]
For these exponents, Zeng et al. \cite{Zeng1} further investigated
the permutation behavior of the polynomials having the form
\[(\Tr_m^n(x)+\delta)^s+L(x)\]
over finite field $\F_{2^n}$, where $m\left|\right.n$ and $L(x)=x$ or $\Tr_m^n(x)+x$, and $``\Tr_m^n(\cdot)"$ is the trace
function from ${\F}_{2^n}$ to ${\F}_{2^m}$ defined by
\[\Tr_m^n(x)=x+x^{2^m}+x^{2^{2m}}+ \cdots + x^{2^{n-m}}.\]
Inspired by the specific PPs obtained by Tu et al. in \cite{Tu1}, some classes of PPs with the form (\ref{1.01})
over $\F_{2^{2m}}$ for exponents $s$ satisfying $s(2^m+1) \equiv 2^m+1$ (mod $2^{2m}-1$)
or $s(2^m-1) \equiv 2^m-1$ (mod $2^{2m}-1$), were presented in \cite{Zha2,Wang}.

Very recently, Gupta and Sharma \cite{Gupta} further investigated  PPs with the form (\ref{1.01}) over $\F_{2^{2m}}$,
where $s=(2^m-1)i+1$ for some $i$ satisfying $5i \equiv 1 ~\textrm{or} ~4~(\textrm{mod} ~2^m+1)$. This result is quite interesting since it gives a much general sufficient condition on permutation property of polynomials  of the form (\ref{1.01}) in the  Niho exponent case. In this paper, we devote to finding more sufficient conditions on the parameter $i$
such that $(x^{2^m}+x+\delta)^{i(2^m-1)+1}+x$ permutes $\F_{2^{2m}}$. One condition we find is  $(2^k+1)i \equiv 1 ~\textrm{or}~ 2^k~(\textrm{mod}~ 2^m+1)$,
where $1 \leq k \leq m-1$  is any integer, which is a direct generalization of Gupta and Sharma's result (the case $k=2$).  However, it turns out that most of previous sporadic constructions of the parameter $i$ in the literature are covered by our result, and of course, it yields many new classes of permutation polynomials as well.

The rest of this paper is organized as follows. In Section \ref{sec2}, some preliminaries needed
 are presented.  In Section \ref{sec3}, we present our main result. Concluding remarks are given in Section \ref{sec4}.

\section{Preliminaries}
\label{sec2}
Let $m$ be an arbitrary positive integer. For each element $x$ in the finite field $\F_{2^{2m}}$, we denote $x^{2^m}$ by $\bar{x}$
in analogy with the usual complex conjugation. Obviously, we have $x+\bar{x} \in \F_{2^m}$ and $x\bar{x} \in \F_{2^m}$.
Define the $unit$ $circle$ of $\F_{2^{2m}}$ by the set
\[U=\{x \in \F_{2^{2m}} : x^{2^m+1}=\bar{x}x=1\}.\]

Over the field $\F_{2^{2m}}$, the so-called $Niho$ $exponent$ \cite{Niho} $s$ has the form $s=(2^m-1)i+2^j$ for two integers $i$ and $j$.
In particular, the exponent $s=(2^m-1)i+1$ is called $normalized$  $Niho$ $exponent$. In this paper, we mainly focus on PPs with form
(\ref{1.01}) over $\F_{2^{2m}}$ with $normalized$  $Niho$ $exponent$ $s$.

We give some lemmas needed in the following sections.


\begin{Lemma} (see \cite[Lemma 2]{Tu1})
\label{Lemma2.2}
Let $a,b$ be two elements of $\F_{2^{2m}}$. Then the equation $\bar{x}+ax+b=0$  has solutions in ${\F}_{2^{2m}}$  as follows:
\begin{align*}
\left\{\begin{array}{ll}
$a unique solution$ ~ \frac{\bar{a}b+\bar{b}}{a\bar{a}+1},   &\mbox{if $a\bar{a} \neq 1$ },\\
2^m ~ solutions,     &\mbox{if $a\bar{a}= 1$ and $\bar{a}b+\bar{b} =0$}, \\
$no solution$,     &\mbox{if $a\bar{a}= 1$ and $\bar{a}b+\bar{b} \neq 0$}.
\end{array}\right.
\end{align*}
\end{Lemma}

\begin{Lemma} (see \cite[Lemma 3]{Tu1})
\label{Lemma2.3}
For any  $\theta  \in \F_{2^{2m}} \setminus \F_{2^m}$, the mapping $\Phi: x \mapsto \frac{\bar{\theta}+z}{\theta+z}$ from
$\F_{2^m}$ to $U\setminus \{1\}$
is  bijective, where $U$ is the unit circle of $\F_{2^{2m}}$.
\end{Lemma}

\begin{Lemma}
\label{Lemma2.4}
Let $k$ be a positive integer, and $x,y$ be  two elements of $\F_{2^{n}}$. Then
\begin{eqnarray}
\label{2.01}
x^{2^k+1}+y^{2^k+1}=(x+y)^{2^k+1}+\sum\limits_{i=0}^{k-1}(xy)^{2^i}(x+y)^{2^k-2^{i+1}+1}.
\end{eqnarray}
\end{Lemma}
\begin{proof}
We prove this result by induction on $k$.

(i) For $k=1$, it is easy to verify that (\ref{2.01}) holds.

(ii) Assuming that (\ref{2.01}) holds for $k >1$, we then prove that it also holds for $k+1$.
Multiplying both sides of  (\ref{2.01}) by $(x+y)^{2^k}$, we get
\begin{eqnarray*}
(x^{2^k+1}+y^{2^k+1})(x+y)^{2^k}=(x+y)^{2^{k+1}+1}+\sum\limits_{i=0}^{k-1}(xy)^{2^i}(x+y)^{2^{k+1}-2^{i+1}+1}.
\end{eqnarray*}
Adding $x^{2^k}y^{2^k}(x+y)$ to both sides of above equation, we have
\begin{eqnarray*}
x^{2^{k+1}+1}+y^{2^{k+1}+1}=(x+y)^{2^{k+1}+1}+\sum\limits_{i=0}^{k}(xy)^{2^i}(x+y)^{2^{k+1}-2^{i+1}+1},
\end{eqnarray*}
that is to say, (\ref{2.01}) holds for $k+1$.
The proof is finished.
\end{proof}

\section{PPs with form (\ref{3.01})}
\label{sec3}
In this section, we investigate the permutation behavior of
\begin{eqnarray}
\label{3.01}
f(x)=(x^{2^m}+x+\delta)^{i(2^m-1)+1}+x
\end{eqnarray}
over $\F_{2^{2m}}$ for any $\delta \in \F_{2^{2m}}$, where $i$ satisfies
$(2^k+1)i\equiv 1~\textrm{or}~ 2^k (\textrm{mod} ~2^m+1)$ for $1 \leq k \leq m-1$.

In order to simplify the proof of the main result of the present paper,  we firstly give some discussions,
which have already  been given by Zha et al. in \cite{Zha2}, for the sake of completeness.

If $\Tr_m^{2m}(\delta)=0$, then $f(x)=(x^{2^m}+x+\delta)^{s}+x$ permutes $\F_{2^{2m}}$ for any exponent $s$.
This is because $(x^{2^m}+x+\delta)^s+x=\gamma$  yields  $(x^{2^m}+x+\delta)^s+\bar{x}=\bar{\gamma}$,
which means $x=(\bar{\gamma}+\gamma+\delta)^s+\gamma$. Therefore,
we  only need to consider the  case $\Tr_m^{2m}(\delta) \neq 0$ in the rest of this paper.

If   $\Tr_m^{2m}(\delta) \neq 0$, to prove  $f(x)$ permutes $\F_{2^{2m}}$,
it suffices to prove that for any $\gamma \in \F_{2^{2m}}$, the equation
\begin{eqnarray}
\label{3.02}
(\bar{x}+x+\delta)^{i(2^m-1)+1}=x+\gamma
\end{eqnarray}
has at most one solution in $\F_{2^{2m}}$.

Let $\theta=\delta+\gamma+\bar{\gamma}$. Then we have $\theta+\bar{\theta}=\delta+\bar{\delta}=\Tr_m^{2m}(\delta) \neq 0$,
and thus $\theta \neq 0$. Hence $x=\gamma$ is not a solution of Eq. (\ref{3.02}).

Raising both sides of Eq. (\ref{3.02}) to the $(2^m+1)$-th power gives
\[(\bar{x}+x+\delta)^{2^m+1}=(x+\gamma)^{2^m+1},\]
which means
\begin{equation}\label{barxxdelta}
\bar{x}+x+\delta=\lambda (x+\gamma)
\end{equation}
for some $\lambda \in U$, i.e.,
\begin{eqnarray}
\label{3.04}
\bar{x}+(1+\lambda)x+\delta+\lambda \gamma=0.
\end{eqnarray}
Substituting Eq. \eqref{barxxdelta} into Eq. (\ref{3.02}), we obtain
\begin{eqnarray}
\label{3.05}
\lambda^{1-2i}(\bar{x}+\bar{\gamma})^i+(x+\gamma)^i=0.
\end{eqnarray}
The  number of solutions of Eq. (\ref{3.04}) can be determined
by Lemma \ref{Lemma2.2}.  We define two sets as
\[U_1=\{\lambda \in U: (1+\bar{\lambda})(1+\lambda)=1\}\]
and
\[U_2=\{\lambda \in U: (1+\bar{\lambda})(1+\lambda) \neq 1\}.\]
$U_1$ and $U_2$ form a disjoint partition of $U$. It can be verified that
\[(1+\bar{\lambda})(1+\lambda)+1=1+\lambda+\bar{\lambda}=\frac{1+\lambda+\lambda^2}{\lambda}.\]
If $\lambda \in U_1$, then $1+\lambda+\lambda^2=0$, i.e., $\lambda^3=1$ and $\lambda \neq 1$.
Note that $\lambda^{2^m+1}=1$ and $(3,2^m+1)=1$ when $m$ is even, therefore, $U_1=\emptyset$ for even $m$.

Now the solutions of Eq. (\ref{3.02}) are divided into the following two cases.

{\bf Case I.} If $\lambda \in U_1$, Eq. (\ref{3.04}) has solutions ($2^m$ solutions) if and only if
\[(1+\bar{\lambda})(\delta+\lambda \gamma)+\bar{\delta}+\bar{\lambda} \bar{\gamma}=0,\]
which implies that $\delta+\bar{\delta}+\bar{\lambda}(\delta+\bar{\gamma}+\gamma)=0$, i.e.,
$\theta+\bar{\theta}+(1+\lambda)\theta=\lambda\theta+\bar{\theta}=0$. Hence, in this case,
there is at most one (none, when $m$ is even) $\lambda \in U_1$  and
it must be equal to the fixed value $\frac{\bar{\theta}}{\theta}$ for a given $\gamma$.

{\bf Case II.} If $\lambda \in U_2$, by Lemma \ref{Lemma2.2}, we can get that Eq. (\ref{3.04}) has one solution
\begin{eqnarray}
\label{3.06}
x=\frac{(1+\bar{\lambda})(\delta+\lambda \gamma)+\bar{\delta}+\bar{\lambda} \bar{\gamma}}{(1+\bar{\lambda})(1+\lambda)+1}.
\end{eqnarray}
Note that $\lambda \in U_2$ is undetermined. Thus, we need  further arguments to fix it. Obviously,  Eq. (\ref{3.06}) leads to
\begin{eqnarray}
\label{3.07}
x+\gamma=\frac{\delta+\bar{\delta}+\bar{\lambda}(\delta+\gamma+\bar{\gamma})}{(1+\bar{\lambda})(1+\lambda)+1}=
\frac{\theta+\bar{\theta}+\bar{\lambda}\theta} {1+\lambda+\bar{\lambda}}
\end{eqnarray}
and $\theta+\bar{\theta}+\bar{\lambda}\theta \neq 0$. Substituting it into Eq. (\ref{3.05}), we obtain a relation between $\lambda$ and $\theta$,

\begin{eqnarray}
\label{3.08}
\lambda^{1-2i}(\theta+\bar{\theta}+\lambda\bar{\theta} )^i+(\theta+\bar{\theta}+\bar{\lambda}\theta )^i=0.
\end{eqnarray}

To summarize, if we wish to prove the permutation property of $f(x)$, it desires to prove that there is at most one
$\lambda \in U$ such that Eq. (\ref{3.02}) has at most one solution in Case I or Case II.

With the above preparations, based on deeper arguments on Case I and Case II,
we investigate the permutation behavior of polynomials with form (\ref{3.01}) in the following.

\subsection{The parity of $m$ and $k$ is different}
\label{subsec3.1}
\begin{Thm}
\label{Thm3.1}
Let $m$  and $k$ be two positive integers with different parity. For any $\delta \in \F_{2^{2m}}$, the polynomial
\[f(x)=(x^{2^m}+x+\delta)^{(2^m-1)i+1}+x\]
permutes $\F_{2^{2m}}$, where $i$ satisfies $(2^k+1)i\equiv 1 ~(mod~2^m+1)$.
\end{Thm}

\begin{proof}
We use the notations and arguments given in the above and only consider the case $\Tr_m^{2m}(\delta) \neq 0$. The discussions are divided into two cases according to the parity of $m$ and $k$.

(\textbf{i}) $m$ is even and $k$ is odd.

Note that $U_1=\emptyset$ under the assumption $m$ is even. So we just
need to consider  case II.

\textbf{Case II.} Assume $\lambda \in U_2=U$. Eq. (\ref{3.08}) means
\[\lambda^{1-2i} (\theta+\bar{\theta}+\lambda \bar{\theta})^{i} = (\theta+\bar{\theta}+\bar{\lambda} \theta)^{i}.\]
Raising both sides of the above equation to the $(2^k+1)$-th power and noting that $(2^k+1)i\equiv 1 ~(\textrm{mod}~2^m+1)$, we obtain
\[\lambda^{2^k-1} (\theta+\bar{\theta}+\lambda \bar{\theta}) =¡¡\theta+\bar{\theta}+\bar{\lambda} \theta,\]
which can be simplified as
\begin{eqnarray}
\label{3.09}
\bar{\theta} \lambda^{2^{k}+1} + (\theta+ \bar{\theta}) \lambda^{2^k} + (\theta+\bar{\theta}) \lambda +\theta=0.
\end{eqnarray}
Now we prove that Eq. (\ref{3.09}) has at most one solution in $U_2$. Suppose there exist $\lambda_1,\lambda_2 \in U_2$ with
$\lambda_1 \neq \lambda_2$ such that

\begin{eqnarray}
\label{3.10}
\bar{\theta} {\lambda_1}^{2^{k}+1} + (\theta+ \bar{\theta}) {\lambda_1}^{2^k} + (\theta+\bar{\theta}) \lambda_1 +\theta=0
\end{eqnarray}
and
\begin{eqnarray}
\label{3.11}
\bar{\theta} {\lambda_2}^{2^{k}+1} + (\theta+ \bar{\theta}) {\lambda_2}^{2^k} + (\theta+\bar{\theta}) \lambda_2 +\theta=0.
\end{eqnarray}
Adding Eq. (\ref{3.10}) and Eq. (\ref{3.11}), we get
\begin{eqnarray*}
\bar{\theta} ({\lambda_1}^{2^{k}+1} + {\lambda_2}^{2^{k}+1}) + (\theta+ \bar{\theta}) ({\lambda_1}^{2^k}+{\lambda_2}^{2^k}) + (\theta+\bar{\theta})(\lambda_1 + \lambda_2) =0.
\end{eqnarray*}
By Lemma \ref{Lemma2.4}, we have
\begin{eqnarray}
\label{3.12}
\bar{\theta} \left( (\lambda_1+\lambda_2)^{2^k+1}+\sum\limits_{i=0}^{k-1}(\lambda_1 \lambda_2)^{2^i}(\lambda_1+\lambda_2)^{2^k-2^{i+1}+1} \right)  \nonumber \\
+ (\theta+ \bar{\theta}) ({\lambda_1}^{2^k}+{\lambda_2}^{2^k}) + (\theta+\bar{\theta})(\lambda_1 + \lambda_2) =0.
\end{eqnarray}
Dividing both sides of Eq. (\ref{3.12}) by $(\lambda_1+\lambda_2)^{2^k+1}$, we get
\begin{eqnarray*}
\bar{\theta} \left( 1+\sum\limits_{i=0}^{k-1}\frac{(\lambda_1 \lambda_2)^{2^i}} { (\lambda_1+\lambda_2)^{2^{i+1}}} \right)
+ (\theta+ \bar{\theta}) \left(\frac{1}{{\lambda_1}+{\lambda_2}} + \frac{1}{(\lambda_1 + \lambda_2)^{2^k}} \right)=0.
\end{eqnarray*}
Substituting $a=\frac{1}{\lambda_1+\lambda_2}$ and $b=a^{2^m}=\frac{\lambda_1 \lambda_2}{\lambda_1+\lambda_2}$ into the above equation,
we have
\begin{eqnarray}
\label{3.13}
\bar{\theta} \left( 1+\sum\limits_{i=0}^{k-1} (ab)^{2^i} \right) + (\theta+ \bar{\theta}) \left(a+a^{2^k} \right)=0.
\end{eqnarray}
Note that $a$ and $b$ may  not belong to $\F_{2^m}$, but $a+b$ and $ab$ do. Adding Eq. (\ref{3.13}) and its $2^m$-th power, we get
\begin{eqnarray*}
(\theta + \bar{\theta}) \left( 1+\sum\limits_{i=0}^{k-1} (ab)^{2^i} \right) + (\theta+ \bar{\theta}) \left(a+b+a^{2^k} +b^{2^k}\right)=0.
\end{eqnarray*}
Since $\theta + \bar{\theta} \neq 0$, this equation leads to
\begin{eqnarray}
\label{3.14}
1+\sum\limits_{i=0}^{k-1} (ab)^{2^i}+a+b+a^{2^k} +b^{2^k}=0.
\end{eqnarray}
Applying $\Tr_1^m(\cdot)$ on  both sides of Eq.(\ref{3.14}), we have
\begin{eqnarray*}
0&=& \Tr_1^m\left(1+\sum\limits_{i=0}^{k-1} (ab)^{2^i}+a+b+a^{2^k} +b^{2^k}\right) \\
&=& \Tr_1^m(1)+\Tr_1^m\left( \sum\limits_{i=0}^{k-1} (ab)^{2^i} \right) + \Tr_1^m(a+b) + \Tr_1^m\left(a^{2^k} +b^{2^k} \right) \\
&=& \Tr_1^m(1)+\Tr_1^m\left( \sum\limits_{i=0}^{k-1} (ab)^{2^i} \right) \\
&=& \Tr_1^m\left( ab \right)  \\
&=& \Tr_1^m\left( \frac{\lambda_1 \lambda_2}{(\lambda_1+\lambda_2)^2} \right)\\
&=&\Tr_1^m\left( \frac{\lambda_1/\lambda_2}{1+(\lambda_1/\lambda_2)^2}\right)\\
&=&\Tr_1^m\left( \frac{1}{1+\lambda_1/\lambda_2}+\frac{1}{1+(\lambda_1/\lambda_2)^2} \right)\\
&=&\sum_{i=0}^{m-1}\left[\left(\frac{1}{1+\lambda_1/\lambda_2}\right)^{2^i} + \left(\frac{1}{1+\lambda_1/\lambda_2}\right)^{2^{i+1}}\right]\\
&=&\frac{1}{1+\lambda_1/\lambda_2} + \frac{1}{1+\lambda_2/\lambda_1}\\
&&(\text{this is because}~(\lambda_1/\lambda_2)^{2^m}=\lambda_2/\lambda_1~\text{since}~\lambda_1/\lambda_2\in U)\\
&=& 1.
\end{eqnarray*}
%
This contradiction implies that Eq. (\ref{3.09})
has at most one solution in $U_2$. Therefore, Eq.(\ref{3.02}) has at most one solution in ${\F}_{2^{2m}}$ for any $\gamma \in {\F}_{2^{2m}}$.

(\textbf{ii}) $m$ is odd and $k$ is even

\textbf{Case I.} Assume $\lambda \in U_1$. From Eq. ({\ref{3.05}}), we can get
\[\lambda^{1-2i}(\bar{x}+\bar{\gamma})^i=(x+\gamma)^i.\]
Since $(2^k+1)i\equiv 1 ~(\textrm{mod}~2^m+1)$, raising to the $(2^k+1)$-th power on both sides leads to
\[\lambda^{2^k-1}(\bar{x}+\bar{\gamma}) = (x+\gamma).\]
Note that $3 \mid (2^k-1)$ ($k$ is even) and $\lambda^3=1$,  the above equation implies  $\bar{x}+x = \bar{\gamma}+\gamma$. Substituting it into Eq .(\ref{3.02}), we get a unique solution
\[x=(\bar{\gamma}+\gamma+\delta)^{(2^m-1)i+1}+\gamma\]
of Eq .(\ref{3.02}). From the discussions before, we note that $\lambda=\frac{\bar{\theta}}{\theta} \in U_1$ is unique.

\textbf{Case II.}
Similar to the proof process for Case II of (\textbf{i}), we can also get  Eq. (\ref{3.14})
\[1+\sum\limits_{i=0}^{k-1} (ab)^{2^i}+a+b+a^{2^k} +b^{2^k}=0.\]
Applying $\Tr_1^m(\cdot)$ on both sides of it, we get
\begin{eqnarray*}
0&=& \Tr_1^m\left(1+\sum\limits_{i=0}^{k-1} (ab)^{2^i}+a+b+a^{2^k} +b^{2^k}\right) \\
&=& \Tr_1^m(1)+\Tr_1^m\left( \sum\limits_{i=0}^{k-1} (ab)^{2^i} \right) + \Tr_1^m(a+b) + \Tr_1^m\left(a^{2^k} +b^{2^k} \right) \\
&=& \Tr_1^m(1)+\Tr_1^m\left( \sum\limits_{i=0}^{k-1} (ab)^{2^i} \right) \\
&=& \Tr_1^m(1) \\
&=& 1.
\end{eqnarray*}
This contradiction implies Eq. (\ref{3.09}) has at most one solution in $U_2$.

We claim that Case I and  Case II can not hold simultaneously.

Similar to Case II of (\textbf{i}), we can get Eq. (\ref{3.09})
\[\bar{\theta} \lambda^{2^{k}+1} + (\theta+ \bar{\theta}) \lambda^{2^k} + (\theta+\bar{\theta}) \lambda +\theta=0.\]
$\lambda =1$ is not a solution of the above equation, since $\theta+\bar{\theta} \neq 0$. By Lemma \ref{Lemma2.3}, each
$\lambda \in U_2 \backslash \{1\}$ can be uniquely expressed as $\lambda=\frac{z+\bar{\theta}}{z+\theta}$ for some $z \in{\F}_{2^{2m}}$.
Substituting it into Eq. (\ref{3.09}) and simplifying, we can get
\begin{eqnarray}
\label{3.15}
(\theta+\bar{\theta})z^{2^k+1}  + (\theta^{2^k} \bar{\theta} + \bar{\theta}^{2^k} \theta) z +(\theta^{2^k}+\bar{\theta}^{2^k})(\theta^2+\theta \bar{\theta} +\bar{\theta}^2)=0.
\end{eqnarray}
Note that Case I holds meaning that $\theta^2+\theta \bar{\theta} +\bar{\theta}^2=0$, further, $\theta^3 +\bar{\theta}^3=0$.
Therefore, $(\theta^{2^k}+\bar{\theta}^{2^k})(\theta^2+\theta \bar{\theta} +\bar{\theta}^2)=0$ and
$\theta^{2^k} \bar{\theta} + \bar{\theta}^{2^k} \theta=\theta \bar{\theta}(\theta^{2^k-1}+ \bar{\theta}^{2^k-1})=
\theta \bar{\theta}\left(\theta^{3 \cdot \frac{2^k-1}{3}}+ \bar{\theta}^{2^k-1}\right)=\theta \bar{\theta}\left(\bar{\theta}^{3 \cdot \frac{2^k-1}{3}}+ \bar{\theta}^{2^k-1}\right)=0$.
Consequently, Eq. (\ref{3.15}) has a unique solution $z=0$, then $\lambda= \frac{\bar{\theta}+z}{\theta+z} =\frac{\bar{\theta}}{\theta} \in U_1$,  which contradicts the
assumption $\lambda \in U_2$.

Summarizing the discussion of (\textbf{ii}),  Eq. (\ref{3.02}) has at most one solution.

Combining (\textbf{i}) with (\textbf{ii}), the proof is finished.
\end{proof}

\begin{Thm}
\label{Thm3.2}
Let $m,t,i,j$ be four positive integers, $1 \leq i,j \leq 2^m$,  $t\cdot i\equiv 1 ~(mod~2^m+1)$ and $t\cdot j \equiv t-1 ~(mod~2^m+1) $.
For any $\delta \in \F_{2^{2m}}$, the polynomial
\[f(x)=(x^{2^m}+x+\delta)^{(2^m-1)i+1}+x\]
permutes $\F_{2^{2m}}$ if and only if
\[g(x)=(x^{2^m}+x+\delta)^{(2^m-1)j+1}+x\]
does.
\end{Thm}
\begin{proof}
Adding $t\cdot i\equiv 1 ~(\textrm{mod}~2^m+1)$ and $t\cdot j \equiv t-1 ~(\textrm{mod}~2^m+1) $,
we get $t(i+j-1)\equiv 0~(\textrm{mod}~2^m+1)$. Note that $\gcd(t,2^m+1)=1$, otherwise, the congruence
equation $t\cdot i\equiv 1 ~(\textrm{mod}~2^m+1)$ has no solution about variable $i$. Therefore,
$2^m+1 \mid i+j-1$, and from $1 \leq i,j \leq 2^m$, we get $i+j=2^m+2$.
\begin{eqnarray*}
 f(x)&=&(\bar{x}+x+\delta)^{(2^m-1)i+1}+x  ~~\textrm{permutes}~~ {\F}_{2^{2m}}  \\
\Leftrightarrow \left(f(\bar{x})\right)^{2^m} &= & (\bar{x}+x+\delta)^{((2^m-1)i+1)2^m}+x   \\
&=& (\bar{x}+x+\delta)^{(2^m-1)(1-i)+1}+x  \\
&=& (\bar{x}+x+\delta)^{(2^m-1)(2^m+2-i)+1}+x  \\
&=& g(x)~~\textrm{permutes}~~ {\F}_{2^{2m}}.
\end{eqnarray*}

The proof is completed.
\end{proof}

\begin{Remark}
\label{Remark3.1} 
Let $m$ be a positive integer and $t=2^k+1$, where $1 \leq k \leq m-1$. 
Let $i$ and $j$ $(1 \leq i,j \leq 2^m)$ be two positive integers satisfying 
$(2^k+1) i\equiv 1 ~(mod~2^m+1)$ and $(2^k+1) j \equiv 2^k ~(mod~2^m+1) $. By Theorem \ref{Thm3.2},
for any $\delta \in \F_{2^{2m}}$, the polynomial
\[f(x)=(x^{2^m}+x+\delta)^{(2^m-1)i+1}+x\]
permutes $\F_{2^{2m}}$ if and only if
\[g(x)=(x^{2^m}+x+\delta)^{(2^m-1)j+1}+x\]
does.
\end{Remark}

As immediate consequences of Theorem \ref{Thm3.1} and Remark \ref{Remark3.1}, we get the following five corollaries.

\begin{Corollary}
\label{Corollary3.1}
Let $m$ be a positive integer. For any $\delta \in {\F}_{2^{2m}}$, both the polynomials
\[f(x)=(x^{2^m}+x+\delta)^{2(2^m-1)+1}+x\]
and
\[g(x)=(x^{2^m}+x+\delta)^{2^m(2^m-1)+1}+x\]
permute ${\F}_{2^{2m}}$.
\end{Corollary}
\begin{proof}
Set $k=m-1$ in Theorem \ref{Thm3.1}. Then $(2^k+1)\cdot 2 \equiv 1~(\textrm{mod}~2^m+1)$, and the parity of $k=m-1$ and $m$ is
obviously different. By Theorem \ref{Thm3.1}, $f(x)$ permutes ${\F}_{2^{2m}}$.

Note that when $k=m-1$, $(2^k+1)\cdot 2^m \equiv 2^k~(\textrm{mod}~2^m+1)$. By the fact that $f(x)$ permutes ${\F}_{2^{2m}}$
and Remark \ref{Remark3.1},  $g(x)$ permutes ${\F}_{2^{2m}}$.
\end{proof}

The results in Corollary \ref{Corollary3.1} were presented in \cite[Theorem 1]{Tu1} and \cite[Theorem3.2(2)]{Wang}.

\begin{Corollary}
\label{Corollary3.2}
Let $m$ be an even positive integer. For any $\delta \in {\F}_{2^{2m}}$, both the polynomials
\[f(x)=(x^{2^m}+x+\delta)^{\frac{2^m+2}{3}(2^m-1)+1}+x\]
and
\[g(x)=(x^{2^m}+x+\delta)^{\frac{2(2^m+2)}{3}(2^m-1)+1}+x\]
permute ${\F}_{2^{2m}}$.
\end{Corollary}
\begin{proof}
We take $k=1$ in Theorem \ref{Thm3.1}, then $(2^k+1)\cdot \frac{2^m+2}{3} \equiv 1~(\textrm{mod}~2^m+1)$, and $k=1$ and $m$
have different parity. By Theorem \ref{Thm3.1}, $f(x)$ permutes ${\F}_{2^{2m}}$.

Note that when $k=1$, $ (2^k+1)\cdot \frac{2(2^m+2)}{3} \equiv 2~(\textrm{mod}~2^m+1)$. From the fact that $f(x)$ permutes ${\F}_{2^{2m}}$
and Remark \ref{Remark3.1}, we get that $g(x)$ permutes ${\F}_{2^{2m}}$.
\end{proof}

The results in Corollary \ref{Corollary3.2} were presented in \cite[Theorem 1]{Zha2} and \cite[Theorem3.2(3)]{Wang}.

\begin{Corollary}
\label{Corollary3.3}
Let $m$ be an odd positive integer, $i$ and $j$ be two positive integers satisfying
$5i \equiv 1 ~({mod}~2^m+1)$ and $5j \equiv  4~({mod}~2^m+1)$, respectively.
For any $\delta \in {\F}_{2^{2m}}$,  both the polynomials
\[f(x)=(x^{2^m}+x+\delta)^{(2^m-1)i+1}+x\]
and
\[g(x)=(x^{2^m}+x+\delta)^{(2^m-1)j+1}+x\]
permute ${\F}_{2^{2m}}$.
\end{Corollary}
\begin{proof}
Set $k=2$ in Theorem \ref{Thm3.1}, and note that the parity of $k$ and $m$ is different. By Theorem \ref{Thm3.1},
$f(x)$ permutes ${\F}_{2^{2m}}$ when $5i \equiv 1 ~({mod}~2^m+1)$. Further, by Remark \ref{Remark3.1},
$g(x)$ permutes ${\F}_{2^{2m}}$ when $5j \equiv 4 ~({mod}~2^m+1)$.
\end{proof}

The results in Corollary \ref{Corollary3.3} were also presented in \cite[Theorem 3.1 and 3.4]{Gupta}.
In fact, the explicit expressions of $i$ and $j$ were obtained in \cite{Gupta}, when $m \equiv  1~(\textrm{mod}~4)$,
$i=\frac{3 \cdot 2^m +4}{5}, j=\frac{2^{m+1}+6}{5}$; when $m \equiv  3~(\textrm{mod}~4)$,
$i=\frac{ 2^m +2}{5}, j=\frac{2^{m+2}+8}{5}$.

\begin{Corollary}
\label{Corollary3.4}
Let $m$ be an even positive integer. For any $\delta \in {\F}_{2^{2m}}$,  both the polynomials
\[f(x)=(x^{2^m}+x+\delta)^{(2^m-1)i+1}+x\]
and
\[g(x)=(x^{2^m}+x+\delta)^{(2^m-1)j+1}+x\]
permute ${\F}_{2^{2m}}$ in either of the following three cases:\\
(\textrm{1}) $m \equiv 0 ~({mod}~3)$ ~~$i=\frac{2^{m+2}+5}{9}$,  $j=\frac{5 \cdot 2^{m}+13}{9}$;\\
(\textrm{2}) $m \equiv 1 ~({mod}~3)$ ~~$i=\frac{2^{m}+2}{9}$,  $j=\frac{ 2^{m+3}+16}{9}$;\\
(\textrm{3}) $m \equiv 2 ~({mod}~3)$ ~~$i=\frac{7 \cdot 2^{m}+8}{9}$,  $j=\frac{ 2^{m+1}+10}{9}$.
\end{Corollary}
\begin{proof}
Note that in either of the above three cases, $9i \equiv 1 ~(\textrm{mod}~2^m+1)$ and $9j \equiv 8 ~(\textrm{mod}~2^m+1)$.
Take $k=3$, and note that the parity of $k$ and $m$ is different. By Theorem \ref{Thm3.1} and Remark \ref{Remark3.1}, the conclusions hold.
\end{proof}

\begin{Corollary}
\label{Corollary3.5}
Let $m$ be an odd positive integer. For any $\delta \in {\F}_{2^{2m}}$,  both the polynomials
\[f(x)=(x^{2^m}+x+\delta)^{(2^m-1)i+1}+x\]
and
\[g(x)=(x^{2^m}+x+\delta)^{(2^m-1)j+1}+x\]
permute ${\F}_{2^{2m}}$ in either of the following four cases:\\
(\textrm{1}) $m \equiv 1 ~({mod}~8)$ ~~$i=\frac{11\cdot 2^{m}+12}{17}$,  $j=\frac{3 \cdot 2^{m+1}+22}{17}$;\\
(\textrm{2}) $m \equiv 3 ~({mod}~8)$ ~~$i=\frac{15\cdot 2^{m}+16}{17}$,  $j=\frac{ 2^{m+1}+18}{17}$;\\
(\textrm{3}) $m \equiv 5 ~({mod}~8)$ ~~$i=\frac{2^{m}+2}{17}$,  $j=\frac{ 2^{m+4}+32}{17}$;\\
(\textrm{4}) $m \equiv 7 ~({mod}~8)$ ~~$i=\frac{5 \cdot 2^{m}+6}{17}$,  $j=\frac{ 3\cdot 2^{m+2}+28}{17}$.
\end{Corollary}
\begin{proof}
Note that in either of the above four cases, $17i \equiv 1 ~(\textrm{mod}~2^m+1)$ and $17j \equiv 16 ~(\textrm{mod}~2^m+1)$.
Take $k=4$ and note that the parity of $k$ and $m$ is different. By Theorem \ref{Thm3.1} and Remark \ref{Remark3.1}, the conclusions hold.
\end{proof}

\subsection{The parity of $m$ and $k$ is same}
When $m$ and $k$ are both  odd integers,  the congruence equation $(2^k+1)\cdot i\equiv 1 ~(\textrm{mod}~2^m+1)$
cannot hold for any $i$, since $\gcd(2^k+1,2^m+1) \neq 1$. Therefore, we  only need to consider the case of both $m$ and $k$
are even.

\begin{Thm}
\label{Thm3.3}
Let $m,k$ be two even positive integers with $k\mid m$ and $\frac{m}{k}$ is even.  For any $\delta \in \F_{2^{2m}}$, the polynomial
\[f(x)=(x^{2^m}+x+\delta)^{(2^m-1)i+1}+x\]
permutes $\F_{2^{2m}}$, where $i$ satisfies $(2^k+1)i\equiv 1 ~(mod~2^m+1)$.
\end{Thm}
\begin{proof}
Similarly to the proof of Theorem \ref{Thm3.1}, we can get the  equation
(Eq. (\ref{3.14}))
\[1+\sum\limits_{i=0}^{k-1} (ab)^{2^i}+a+b+a^{2^k} +b^{2^k}=0.\]
Applying $\Tr_k^m(\cdot)$ on both sides of the  equation, we get
\begin{eqnarray*}
0&=& \Tr_k^m\left(1+\sum\limits_{i=0}^{k-1} (ab)^{2^i}+a+b+a^{2^k} +b^{2^k}\right) \\
&=& \Tr_k^m(1)+\Tr_k^m\left( \sum\limits_{i=0}^{k-1} (ab)^{2^i} \right) + \Tr_k^m(a+b) + \Tr_k^m\left(a^{2^k} +b^{2^k} \right) \\
&=&  \Tr_k^m\left( \sum\limits_{i=0}^{k-1} (ab)^{2^i} \right) \\
&=&  \Tr_1^m(ab) \\
&=& \Tr_1^m\left( \frac{\lambda_1 \lambda_2}{(\lambda_1+\lambda_2)^2} \right)\\
&=& 1.
\end{eqnarray*}
The remainder of this proof is similar to the proof of Theorem \ref{Thm3.1}, we omit the details here.
\end{proof}


As immediate consequences of Theorem \ref{Thm3.3} and Remark \ref{Remark3.1}, we get the following two corollaries.

\begin{Corollary}
\label{Corollary3.6}
Let $4 \mid m$. For any $\delta \in {\F}_{2^{2m}}$,  both the polynomials
\[f(x)=(x^{2^m}+x+\delta)^{\frac{2^{m+1}+3}{5}(2^m-1)+1}+x\]
and
\[g(x)=(x^{2^m}+x+\delta)^{\frac{3\cdot 2^m+7}{5}(2^m-1)+1}+x\]
permute ${\F}_{2^{2m}}$.
\end{Corollary}
\begin{proof}
Let $k=2$. Then $m/k$ is even, since $4 \mid m$. Note that $(2^k+1) \frac{2^{m+1}+3}{5} \equiv 1 ~(\textrm{mod}~2^m+1)$,
$(2^k+1)  \frac{3\cdot 2^m+7}{5} \equiv 2^k ~(\textrm{mod}~2^m+1)$.
By Theorem \ref{Thm3.3} and Remark \ref{Remark3.1}, the conclusions hold.
\end{proof}

\begin{Corollary}
\label{Corollary3.7}
Let $8 \mid m$. For any $\delta \in {\F}_{2^{2m}}$,  both the polynomials
\[f(x)=(x^{2^m}+x+\delta)^{\frac{2^{m+3}+9}{17}(2^m-1)+1}+x\]
and
\[g(x)=(x^{2^m}+x+\delta)^{\frac{ 9\cdot 2^m+25}{17}(2^m-1)+1}+x\]
permute ${\F}_{2^{2m}}$.
\end{Corollary}
\begin{proof}
Let $k=4$, then $\frac{m}{k}$ is even since $8 \mid m$. Note that  $(2^k+1) \frac{2^{m+3}+9}{17} \equiv 1 ~(\textrm{mod}~2^m+1)$,
$(2^k+1)  \frac{ 9\cdot 2^m+25}{17} \equiv 2^k ~(\textrm{mod}~2^m+1)$.
By Theorem \ref{Thm3.3} and Remark \ref{Remark3.1}, the conclusions hold.
\end{proof}


\section{Concluding remarks}
\label{sec4}
In this paper, we investigate PPs of the form $(x^{2^m}+x+\delta)^{i(2^m-1)+1}+x$ over the finite field $\F_{2^{2m}}$ for
$i$ satisfying $(2^k+1)i \equiv 1 ~\textrm{or}~ 2^k ~(\textrm{mod}~ 2^m+1)$, where $1 \leq k \leq m-1$.
Most of the previous constructions  can be covered by our results. Besides,
many new classes of PPs with the above form are obtained. Most importantly, we can easily get
infinitely many classes of PPs with the form (\ref{3.01}) only through solving the
congruence equation $(2^k+1)\cdot i\equiv 1 ~\textrm{or}~ 2^k ~(\textrm{mod}~2^m+1)$  under some restriction on $k$ and $m$.

We summarize all the known PPs of the form $f(x)=(x^{2^m}+x+\delta)^{(2^m-1)i+1}+x$ over $\F_{2^{2m}}$ in Table \ref{table1}, in which
only the last two rows  are not  covered by our result. At last,
we again emphasis that there is no restriction on $\delta \in \F_{2^{2m}}$ in all the classes.

\scriptsize
\begin{table}[!h]\caption{{ Known PPs of the form $(x^{2^m}+x+\delta)^{i(2^m-1)+1}+x$ over $\F_{2^{2m}}$}}
	\label{table1}
	\begin{center}
		\begin{tabular}{|c|c|c|c|}
			\hline
			$m$¡¡     &¡¡$k$  &     $i$  &  References \\
			\hline
			all $m$   &   $m-1$ &    $i \in \{2,2^m\}$    &  \cite{Tu1} \cite{Wang}, Corollary \ref{Corollary3.1} \\
			\hline
			$m$ is even  &  1   &   $i \in \left\{\frac{2^m+2}{3},   \frac{2\cdot (2^m+2)}{3} \right\}$ &  \cite{Zha2} \cite{Wang}, Corollary \ref{Corollary3.2} \\
			\hline
			$m$ is odd  &   2  &
			\begin{tabular}{c}
				$m \equiv 1 ~(\textrm{mod}~ 4), i \in \{\frac{3 \cdot 2^{m} +4}{5},  \frac{ 2^{m+1} +6}{5}\}$ \\
				$m \equiv 3 ~(\textrm{mod}~ 4), i \in \{\frac{ 2^{m} +2}{5},  \frac{ 2^{m+2} +8}{5}\}$
			\end{tabular}
			& \cite{Gupta}, Corollary \ref{Corollary3.3}   \\
			\hline
			$m$ is even  &   3  &
			\begin{tabular}{c}
				$m \equiv 0 ~({mod}~3),~i\in \left\{\frac{2^{m+2}+5}{9},  \frac{5 \cdot 2^{m}+13}{9}\right\}$ \\
				$m \equiv 1 ~({mod}~3),~i\in \left\{\frac{2^{m}+2}{9}, \frac{ 2^{m+3}+16}{9}\right\}$\\
				$m \equiv 2 ~({mod}~3),~i\in \left\{\frac{7 \cdot 2^{m}+8}{9}, \frac{ 2^{m+1}+10}{9}\right\}$
			\end{tabular}
			& Corollary \ref{Corollary3.4}   \\
			\hline
			$m$ is odd  &   4  &
			\begin{tabular}{c}
				$m \equiv 1 ~({mod}~8), i\in \left\{\frac{11\cdot 2^{m}+12}{17}, \frac{3 \cdot 2^{m+1}+22}{17}\right\}$\\
				$m \equiv 3 ~({mod}~8), i\in \left\{\frac{15\cdot 2^{m}+16}{17}, \frac{ 2^{m+1}+18}{17}\right\}$\\
				$m \equiv 5 ~({mod}~8), i\in \left\{\frac{2^{m}+2}{17}, \frac{ 2^{m+4}+32}{17}\right\}$\\
				$m \equiv 7 ~({mod}~8), i\in \left\{\frac{5 \cdot 2^{m}+6}{17}, \frac{ 3\cdot 2^{m+2}+28}{17}\right\}$
			\end{tabular}
			& Corollary \ref{Corollary3.5}   \\
			\hline
			$4 \mid m$  &   2  &  $i \in \left\{\frac{2^{m+1}+3}{5},   \frac{3 \cdot 2^m+7}{5} \right\}$   & Corollary \ref{Corollary3.6}   \\
			\hline
			$8 \mid m$  &   4  &  $i \in \left\{\frac{2^{m+3}+9}{17},   \frac{9 \cdot 2^m +25}{17} \right\}$   & Corollary \ref{Corollary3.7}   \\
			\hline
			all $m$  &    &  $i =2^{m-1}+1$   & \cite{Wang}   \\
			\hline
			$m \not\equiv 0~(\textrm{mod}~3)$ &    &  $i \in \{2^{m-2}+1, 3\cdot 2^{m-2}+1\}$   & \cite{Wang}  \cite{Zha2} \\
			\hline
		\end{tabular}
	\end{center}
 $k$ satisfies $ (2^k+1)i \equiv 1 ~\textrm{or}~ 2^k ~(\textrm{mod}~ 2^m+1)$
\end{table}\normalsize

\section*{Acknowledgements}
This work is partially supported by  the National Natural Science Foundation of China
under Grant No. 61502482.

\end{document}